\documentclass[letterpaper, 10 pt, conference]{ieeeconf}  

\IEEEoverridecommandlockouts                              
\usepackage{syntax}
\usepackage{url}
\usepackage{amssymb}

\usepackage{amsfonts}
\usepackage[english]{babel}
\usepackage[bottom]{footmisc}
\usepackage{algorithm}
\usepackage{algpseudocode}
\usepackage{amsthm}
\usepackage{siunitx}
\usepackage{subcaption}
\usepackage{amsmath}
\usepackage{adjustbox}
\usepackage{amsmath}
\usepackage{xcolor}
\usepackage{cite}
\usepackage{amsfonts}
\usepackage{xspace} 
\usepackage{comment}

\usepackage{eso-pic}

\graphicspath{{figures/}}

\newtheorem{theorem}{Theorem}

\DeclareMathOperator*{\argmin}{arg\,min}

\newcommand{\strategy}{\omega}

\renewcommand{\wp}{
\text{w}}
\newcommand{\wpSet}{W}

\newcommand{\costMDP}{C}
\newcommand{\tol}{\varepsilon}

\newcommand{\reals}{\mathbb{R}}
\newcommand{\naturals}{\mathbb{N}}
\newcommand{\beliefs}{\mathbb{B}}

\newcommand{\I}{\mathcal{I}}
\renewcommand{\P}{\mathcal{P}}
\newcommand{\prob}{\mathrm{P}}
\newcommand{\coll}{\text{coll}}
\newcommand{\tar}{\text{tar}}

\title{\LARGE \bf Pareto Optimal Strategies for Event Triggered Estimation}

\author{Anne Theurkauf, Nisar Ahmed, and Morteza Lahijanian
\thanks{Authors are with the Smead Department of Aerospace Engineering Sciences, University of Colorado Boulder, CO, USA.
        {\tt \small \{firstname.lastname\}@colorado.edu}
        }%
}

\newcommand{\ml}[1]{\textcolor{blue}{[ML: #1]}}
\newcommand{\na}[1]{{\color{red}[NA: #1]}}
\newcommand{\at}[1]{{\color{purple}[AT: #1]}}

\begin{document}

\AddToShipoutPictureBG*{%
  \AtPageUpperLeft{%
    \hspace{16.5cm}%
    \raisebox{-1.5cm}{%
      \makebox[0pt][r]{To appear in the IEEE Conference on Decision and Control (CDC), December 2022.}}}}

\maketitle

\thispagestyle{plain}
\pagestyle{plain}

\begin{abstract}
Although resource-limited networked autonomous systems must be able to efficiently and effectively accomplish tasks, better conservation of resources often results in worse task performance. We specifically address the problem of finding strategies for managing measurement communication costs between agents. A well understood technique for trading off communication costs with estimation accuracy is event triggering (ET), where measurements are only communicated when useful, e.g., when Kalman filter innovations exceed some threshold. In the absence of measurements, agents can use implicit information to achieve performance almost as good as when explicit data is always communicated. However, there are no methods for setting this threshold with formal guarantees on task performance. We fill this gap by developing a novel belief space discretization technique to abstract a continuous space dynamics model for ET estimation to a discrete Markov decision process, which scalably accommodates threshold-sensitive ET estimator error covariances. We then apply an existing probabilistic trade-off analysis tool to find the set of all optimal trade-offs between resource consumption and task performance. From this set, an ET threshold selection strategy is extracted. Simulated results show our approach identifies non-trivial trade-offs between performance and energy savings, with only modest computational effort. 
\end{abstract}

\section{INTRODUCTION}
    \label{sec:intro}
    
As autonomous agents become more prevalent in society, it is increasingly important to find ways to efficiently and effectively deploy them. These systems must be able to reliably perform given tasks with limited resources. This is a difficult challenge given that these two aspects are naturally competing: often the best performance of a task comes at a greater expense of resources. One facet of this problem is the cost of accurate estimation of the system state; more measurements of the state leads to better estimation and hence better decisions, but this comes at some resource expense. This is a particularly troublesome issue for networked systems that rely on shared information, where communication cost competes with the accuracy of state estimation. 

For example, a system might combine mobile agents with remote sensors as shown in Figure \ref{fig:EstimationDiagram}. While the sensors can make and communicate measurements with the agents, they may have limited computing capabilities and battery life and, in remote environments, may be difficult to access to change batteries. According to \cite{Shi2014:ML}, wireless transmission is at least an order of magnitude more costly in terms of energy consumption than any other function in a standard ZigBee networking chip. Maintaining consistent operation in these kinds of scenarios necessarily means limiting  communication in the remote sensing network. Additionally, systems may be bandwidth limited, in which case limiting communication frees up space for other valuable information, e.g., science data in robotic exploration missions. In this paper, we develop a formal framework to explore optimal trade-offs between communication cost and task performance. 

A popular method of managing communication costs is event-triggered (ET) estimation \cite{Trimpe2015, Ouimet2018, Shi2014, Wu2013}. In the ET framework, information is only communicated when it is deemed useful. This condition can take many forms \cite{Trimpe2015}, but a common choice is a threshold on the innovation of a measurement \cite{Ouimet2018, Shi2014, Wu2013}. By changing this threshold, we can change the frequency of communication. While this can be a powerful tool for trading-off energy usage with estimator accuracy, there are no methods for tuning this threshold to guarantee the  trade-off is optimal and accounts for task performance. 

Trade-off analysis techniques have been studied mostly in formal methods literature \cite{Etessami2007,Forejt2012, Hahn:QEST:2017, Hahn:TOMACS:2019}.  
Given a set of quantitative objectives and system model, the goal of multi-objective analysis is to calculate the set of all optimal trade-offs between the objectives, called the Pareto front.  Recent studies explore applications in robotics \cite{Romanova2020}, path planning \cite{lavin2015pareto,Lee2018}, multi-agent coordination \cite{Zhao2021,Ghrist2006,Cui2012}, 
and sensor scheduling \cite{Lahijanian2018}.
The latter specifically focuses on resource and performance objectives and develops a scheduling framework for the usage of a high precision, high cost sensor. Nevertheless, finding Pareto-optimal communication strategies between a sensor network and mobile agents with respect to multiples objectives remains an open problem.

\begin{figure}
    \centering
    \includegraphics[width=0.45\textwidth]{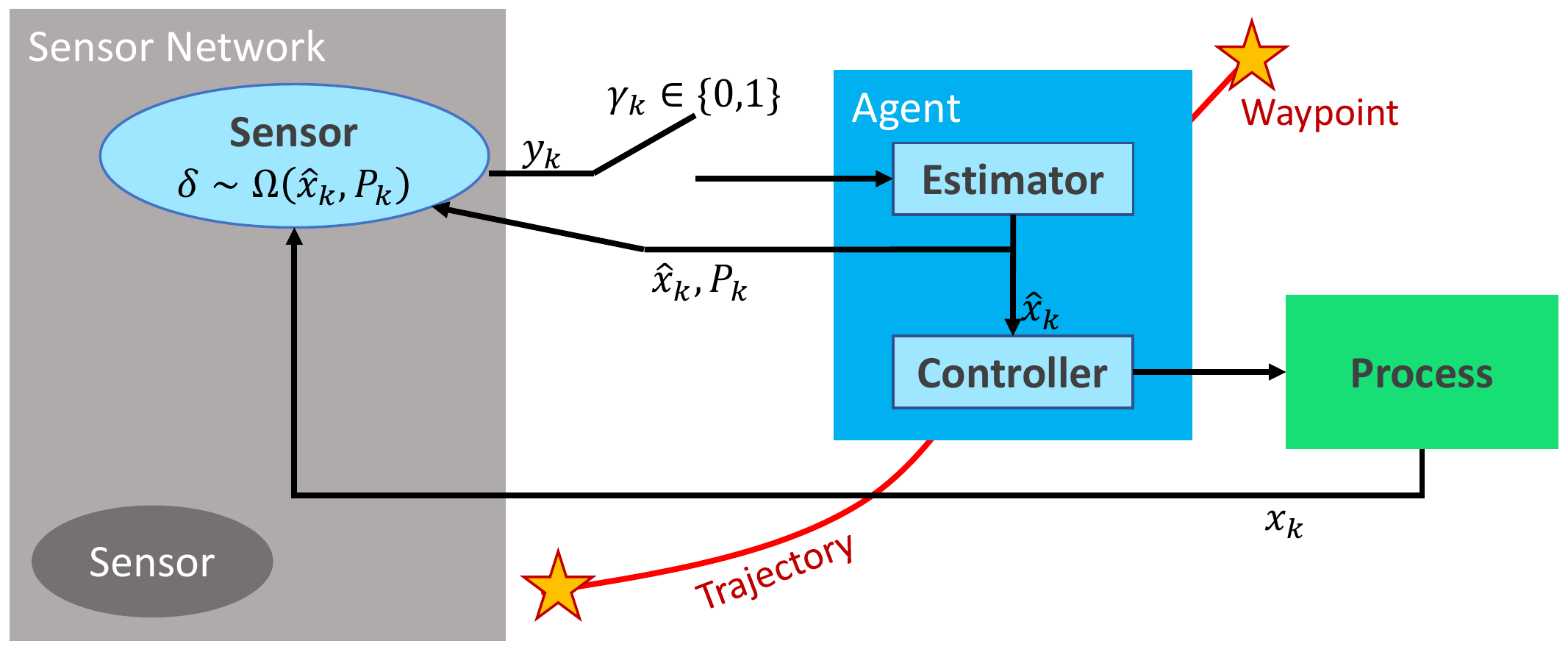}
    \vspace{-2mm}
    \caption{Diagram of triggering and estimation process}
    \label{fig:EstimationDiagram}
    \vspace{0mm}
\end{figure}

In this paper, we fill this gap by introducing a method of synthesizing strategies for setting ET thresholds to achieve optimal trade-offs between resource consumption and task performance. We specifically consider a scenario where a mobile agent has to reach a goal region while avoiding unsafe states (obstacles) by communicating with a resource-limited remote sensor network. 
Our approach is based on abstracting this system to a discrete Markov Decision Process (MDP). The abstraction involves a novel belief space discretization method  
using the spectral decomposition of the estimated covariance. 
We show this abstraction method is more efficient than alternatives and also benefits from interpretability. 
Then, we leverage off-the-shelf tools (PRISM \cite{Kwiatkowska2011}, PRISM-games \cite{Kwiatkowska2013}) to obtain the Pareto front and synthesize ET strategies for individual points on the front.  We show the efficacy of the approach in several case studies.

In short, the contributions are: (i) the use of multi-objective analysis to synthesize ET threshold strategies with optimal resource-performance trade-off guarantees, (ii) a novel technique to discretize a Gaussian belief based on the spectral decomposition of its covariance matrix, (iii) a method of using this discretization to abstract the continuous space system to an MDP, and (iv) a set of case studies that illustrate the efficacy of the method.


\section{PROBLEM FORMULATION}
    \label{sec:problem}

We consider an active agent communicating with a resource-limited remote sensor network. 
The nearest sensor to the agent is able to make and communicate measurements of the agent's state to improve its performance, but the communication comes at a resource cost. 
This setup can apply to many types of measurements, e.g.  remote one-way ranging measurements from beacons whose battery energy and data volume must be conserved in austere environments. 
The goal is to find a communication strategy to optimally trade-off sensor resource costs and agent task performance.

\subsection{System Model}
\label{subsec:SysModel}
The active agent evolves according to linear dynamics, 
\begin{equation}
    \label{eq:system dynamics}
    x_{k+1} = F x_k + G u_k + w_k, \quad w_k \sim \mathcal{N}(0,Q),
\end{equation}
where $x_k\in X \subseteq \mathbb{R}^n$ and $u_k\in U\subseteq\mathbb{R}^p$ are the state and input, respectively, with $F\in\mathbb{R}^{n\times n}$ and $G\in\mathbb{R}^{n\times p}$. Random process noise $w_k \in \reals^n$ is modeled as white Gaussian noise with 
covariance $Q\in\mathbb{R}^{n\times n}$.
The initial agent state $x_0 \sim \mathcal{N}(\hat{x}_0,P_0)$ is assumed Gaussian distributed, with mean $\hat{x}_0 \in \mathbb{R}^n$ and covariance $P_0 \in \mathbb{R}^{n \times n}$. 
%
The \emph{nearest} remote sensor measures the agent state as
\begin{equation}
    \label{eq:system measurement}
    y_k = Hx_k + v_k,\quad v_k \sim \mathcal{N}(0,R),
\end{equation}
where $y_k\in Y\subseteq\mathbb{R}^m$ is the measurement with $H\in\mathbb{R}^{m\times n}$. Random measurement noise $v_k \in \reals^m$ is also modeled as white Gaussian noise with 
covariance $R\in\mathbb{R}^{m\times m}$. 
We assume the system is both \emph{controllable} and \emph{observable}, and that the noise 
covariance matrices $Q$ and $R$ are positive definite.
For ease of presentation, we assume the agent 
only has access to measurement data from the nearest remote sensor. The extension to include local measurements is trivial.

\subsection{Event-Triggered State Estimation}
\label{subsec:etest}

Since sensors are resource limited, the agent uses an event triggered (ET) state estimator for efficient information sharing. 
When a measurement $y_k$ is taken, the sensor determines whether to report $y_k$ to the agent. This decision is based on a shared estimate of $x_k$ provided by some central estimator, e.g., the agent.  
We specifically focus on innovation-based ET described in \cite{Wu2013} although many other triggers exist \cite{Trimpe2015}. 

Let $\gamma_k \in \{0,1\}$ be the indicator for measurement communication at time $k$, where $\gamma_k = 1$ if the measurement is sent, and $\gamma_k = 0$ otherwise. Define $\I_k = (I_1, \ldots, I_k)$, $k \in \naturals$, as the sequence of available \emph{information} up to time $k$, where 
\begin{equation}
    I_j = \begin{cases}
    \{y_j\} & \text{if $\gamma_j=1$}\\
    \emptyset & \text{otherwise.}
    \end{cases}, \ \mbox{for } 1 \leq j \leq k.
\end{equation}
Information $\I_k$ is used to reason over the probability distribution of $x_k$.
This distribution is denoted by $b_k$ and is referred to as the \emph{belief}, i.e., 
\begin{equation}
    x_k \sim b_k = \P(x_k \mid x_0, \mathcal{I}_k), \label{eq:beldef}
\end{equation}
and the set of all beliefs is denoted by $\beliefs$. 

Under the assumption that $\P(x_k \mid x_0, \mathcal{I}_{k-1})$ is Gaussian, work \cite{Wu2013} introduces an ET Minimum Mean Square Error (MMSE) estimator,
where the belief is fully described by its mean and covariance.
The expected value of the state based on the belief is the MMSE \textit{state estimate}. Before a measurement occurs, the \emph{a priori} (predicted) state estimate and error $e_k^-$ covariance at time $k$ are given by
\begin{align}
    &\hat{x}_k^-=\mathbb{E}[x_k \mid \mathcal{I}_{k-1}], \quad 
    e_k^{-}=x_k - \hat{x}_k^{-}, \\ 
     &P_k^{-} = \mathbb{E}[e_k^{-} e^{-T}_k \mid \mathcal{I}_{k-1}].
\end{align}
After a measurement occurs, the \emph{a posteriori} (updated) estimate and error $e_k$ covariance are, for any choice of $\gamma_k$,
\begin{align}
    &\hat{x}_k = \mathbb{E}[x_k \mid \mathcal{I}_k], \label{eq:postmeandef} \quad
     e_k = x_k - \hat{x}_k, \\
    &P_k = \mathbb{E}[e_k e_k^T \mid \mathcal{I}_k] \label{eq:postcovdef}.
\end{align}
In innovation-based ET, $\gamma_k$ is determined based on the measurement innovation,
$z_k = y_k - \mathbb{E}[y_k \mid \mathcal{I}_{k-1}],$
which is a Gaussian random vector with covariance $Z_k = R + HP_k^- H^T$. 
Let $\mathcal{F}_k = V_k \Lambda^{-1/2}_k$, where the columns of $V_k \in \reals^{n\times n}$ are the orthonormal eigenvectors of $Z_k$, and $\Lambda_k\in\mathbb{R}^{n\times n}$ is the diagonal eigenvalue matrix of $Z_k$, 
so that the covariance of $\epsilon_k = \mathcal{F}_k z_k$ is the identity matrix. The triggering condition is then defined in terms of a given threshold $\delta \in \mathbb{R}_{>0}$ as, 
\begin{equation}
    \gamma_k = 
    \begin{cases}
    0 & \text{if $\|\epsilon_k\|_\infty\leq\delta$}\\
    1 & \text{otherwise}.
    \end{cases}
\end{equation}
This corresponds to assessing how `surprising' $y_k$ is. For small $\delta$, many measurements are surprising, and many triggers (communications) occur; for large $\delta$, fewer measurements are surprising and thus fewer are communicated. 
As such, (\ref{eq:postmeandef}) and (\ref{eq:postcovdef}) can be obtained through a recursive Kalman filter approximation, which attenuates the Kalman gain as a function of $\delta$ when $\gamma_k=0$; see \cite{Wu2013} for complete details of this ET Kalman filter, which produces a standard Kalman update if $\gamma_k = 1$ or $\delta=0$. 
The key advantage of the ET filter is when $\gamma_k = 0$ (i.e., $I_k = \emptyset$), since the agent can still extract `implicit' information from the fact that $\|\epsilon_k\|_\infty\leq\delta$. This can be used to compute a more accurate \textit{a posteriori} estimate than the one obtained simply by keeping an unmodified \emph{a priori} estimate from the Kalman prediction step. 

Under ET estimation, the agent's task performance and sensor resource consumption now depend on the choice of $\delta$. 
However, there are no established methods to algorithmically select or adjust $\delta$ to provide collective formal guarantees on task performance and resource consumption effects. 
In this work, we design a strategy for choosing ET thresholds 
over the course of the agent's trajectory to achieve an optimal resource-performance trade-off.

\subsection{Task and Control Laws}
\label{subsec:controlLaws}

The agent's task is to follow a nominal trajectory to a target region while avoiding unsafe states (e.g. obstacles). 
Assume the trajectory is provided via a sequence of waypoints $\wpSet = (\wp_0, \ldots, \wp_N)$ such that at $\wp_i \in \mathbb{R}^n$, for each $i\in \{0,..,N\}$, the nearest sensor to the agent changes. Such waypoints can be easily obtained, e.g. using a motion planner followed by a Voronoi segmentation with respect to the sensor network. 
Hence, the waypoints are the decision points where the triggering threshold for the next sensor (and thus the corresponding trajectory segment) is determined. 

To follow the trajectory, the agent is equipped with a series of (feedback) control laws connecting the waypoints in $\wpSet$, coupled with a termination condition that determines when to switch to the next controller.
Let $\mathcal{U}_i: X \times \wpSet \to U$ be the control law that drives the agent from waypoint $\wp_{i-1}$ to $\wp_{i}$ and $\zeta_i: X \times \wpSet \times \naturals \to \{0,1\}$ be the corresponding termination condition.  Then, the agent applies control law $\mathcal{U}_i$ until $\zeta_i$ is triggered.  We design $\zeta_i$ to trigger as  
\begin{equation}
    \label{eq:meanConvCond}
    \zeta_i(\hat{x}_k, \wp_i, k_i) = 
    \begin{cases}
        1 & \text{if } \|\hat{x}_k - \wp_i\|_2 \leq \tol_x \text{ or } k_{i} \geq k_{\max}\\
        0 & \text{otherwise,}
    \end{cases}
\end{equation}
where $\tol_x \in\mathbb{R}_{>0}$ is the convergence tolerance, $k_{i}$ is time duration since switching to $\mathcal{U}_i$, and $k_{\max}$ is the max time duration threshold. The only requirement on $\mathcal{U}_i$ is to drive the agent to proximity of $\wp_i$. It can, for instance, consist of a reachability and a stabilizing controller using the LQG control architecture described in \cite{AghaMohammadi2014,Lahijanian2018}.

\subsection{Resource and Performance Objectives}
\label{subsec:objFunc}
We seek to choose ET thresholds at waypoints in $\wpSet$ to achieve an optimal resource-performance trade-off.
Let $\Delta = \{\delta_1,\delta_2,\ldots, \delta_{|\Delta|}\}$ be a set of ET threshold candidates, and $\mathcal{D}(\Delta)$ be the set of all probability distributions over $\Delta$.
We define an \textit{ET strategy} $\strategy: \mathbb{B} \rightarrow \mathcal{D}(\Delta)$ to be a function that maps a belief $b_k \in \mathbb{B}$ at a waypoint to a probability distribution over $\Delta$.
We want to compute $\strategy$ according to the following three competing objectives: 
(O1) minimize the probability of visiting an unsafe state (obstacle), (O2) maximize the probability of reaching the target region, and (O3) minimize resource consumption.  We refer to (O1) and (O2) as the performance objective and (O3) as the resource objective. We now formalize (O1)-(O3) as a function of $\strategy$.

Given $\strategy$, $b_k$ itself is a random variable before agent deployment due to the stochasticity of $y_k$.
To provide guarantees prior to deployment, we express the objectives based on the \emph{expected belief} defined as:
\begin{align}
    \label{eq: expected belief}
    \textbf{b}_k(\strategy) &= \mathbb{E}_Y(b_k \mid x_0, y_{0:k}, \strategy) \\
    &= \int_{y_{0:k}} \P(x_k \mid x_0, y_{0:k}, \strategy) pr(y_{0:k})dy \nonumber \\
    &= \P(x_k \mid x_0, \strategy) \nonumber
\end{align}

Consider a set of $M$ obstacles (unsafe states) $\mathcal{O}=\{X_{O_1}, ..., X_{O_M}\}$, where $X_{O_i} \subset X$ for every $i \in \{1,...,M\}$. 
Assuming that collisions with obstacles are terminal, the probability of collision with obstacle $X_{O_i}$ under $\strategy$ is 
\begin{equation*}
    \prob_{\coll,i}(\strategy) = pr(x_{0:k_T} \in X_{O_i}  \mid \strategy) = \int_{X_{O_i}} \hspace{-2mm} \textbf{b}_{k_T} (\strategy) dx_{k_T},
\end{equation*}
where $k_T$ is the termination time of traversing trajectory $\wpSet$.
Then, for (O1), the total probability of collision is the sum of the probabilities over all the individual obstacles,
\begin{equation}
    \prob_{\coll}(\strategy) = pr(x_{0:k_T} \in \mathcal{O} \mid \strategy) = \sum_{i=1}^{M} \prob_{\coll,i}(\strategy).
\end{equation}
Similarly, for (O2), the probability of ending in the target region $X_{\tar} \subset X$ is given by
\begin{equation}
    \prob_{\tar}(\strategy) = pr(x_{k_T} \in X_{\tar} \mid \strategy) = \int_{X_{\tar}} \hspace{-2mm} \textbf{b}_{k_{T}} (\strategy) dx_{k_T}.
\end{equation}


For (O3), the resource cost can generally be defined as a function of the $x$, $u$, and $\delta$. For ease of presentation, we define this cost solely based on the energy consumed by communicating measurements, but the extension to include controls is straightforward. 
The instantaneous communication cost at time $k$ is $c_m\gamma_k$, where $c_m \in \reals_{\geq 0}$. The expected $\gamma_k$ is a function of $\delta$, i.e. $\bar{\gamma}_k(\delta) = \mathbb{E}[\gamma_k]$. 
Then, the expected total cost over a trajectory is
\begin{equation}
    C_E(\strategy) = \sum_{k=0}^{k_T}  \sum_{\delta \in \Delta} \strategy(\textbf{b}_k)(\delta) \, c_m\bar{\gamma}_k(\delta).
\end{equation}


\subsection{Problem Statement}
We consider the following problem: 
given a system model as described in Section \ref{subsec:SysModel}, a set of obstacles $\mathcal{O}$, a target region $X_{\tar}$, waypoints $\wpSet$, controllers as defined in Section \ref{subsec:controlLaws}, and a set of ET thresholds $\Delta$, compute optimal ET strategy $\strategy^*$ such that
\begin{equation}
    \label{eq: pareto optimal}
    \strategy^* = \argmin_{\strategy}\big(C_E(\strategy), \prob_{\coll}(\strategy), 1-\prob_{\tar}(\strategy) \big),
\end{equation}
where $\min$ is a simultaneous minimization of every element. 
%

Note that \eqref{eq: pareto optimal} is a multi-objective optimization problem, and since the objectives are (often) competing, there may not exist a single solution that simultaneously optimizes each objective.  For this reason, we study the optimal trade-offs between the objectives.  Specifically, we aim to find the set of all optimal trade-offs  (Pareto front). Then, given a point from this set, we synthesize the corresponding $\strategy^*$.


It is exceedingly difficult to directly calculate the set of all trade-offs for a continuous system.  We instead approach this problem by constructing an abstraction of the continuous system as a finite MDP.  Multi-objective optimization algorithms are well established for MDPs and implemented in tools such as PRISM \cite{Kwiatkowska2011}, which we use for our analysis.


\section{ABSTRACTION METHOD}
    \label{sec:Methodologies}

Finite abstraction is often performed via a discretization of the continuous state space.
In absence of sensing uncertainty, this is easily addressed by polytopic partitions, e.g., a grid. However, with measurement noise, the true state of the system is unknown; hence, we must reason over the belief space. 
Recall from \eqref{eq:beldef} that the belief at a waypoint is dependent on the history of measurements and threshold choice. Consequently, the belief evolution captured at waypoints is as an exponentially growing graph \cite{Lahijanian2018}. 
To avoid such explosion, we seek discretization methods that break the history dependence.
Specifically, 
we achieve this by either enforcing convergence to a single belief at a waypoint or by grouping the beliefs at each waypoint.  Then, the transition from a belief $b$ at waypoint $\wp_i$
to another belief at $\wp_{i+1}$ 
becomes dependent only on $b$ and the choice $\delta$ at $\wp_i$, not on the entire history. This naturally leads to an MDP construction where the MDP states are formed by the converged belief or belief sets.


Existing methods, like \cite{AghaMohammadi2014} and \cite{Lahijanian2018}, rely on convergence to a known belief. Our first proposed method is an arguably trivial extension of this existing work, wherein we enforce convergence to steady state beliefs about the waypoints. However, we demonstrate that this naive approach yields too coarse a gradation in the Pareto front, and does not fully reap the energy savings offered by ET. To address this, we develop a novel abstraction method that exploits the inherent geometry of the variable covariance hyperellipsoids obtained under ET to discretize and group similar beliefs together.

\subsection{MDP Abstraction}
An MDP $\mathcal{M}=(S, s_0, A, T, \costMDP)$ is a tuple consisting of a finite set of states $S$, an initial state $s_0 \in S$, a set of actions $A$, a probabilistic transition function $T: S \times A \rightarrow [0,1]$, and a cost function $\costMDP: S \times A \rightarrow \mathbb{R}_{\geq0}$ that assigns to each state-action pair to a non-negative cost.

In our MDP abstraction, the action set $A = \Delta$, the set of ET thresholds, and each state $s \in S$ is defined as a set of beliefs at each waypoint. 
The cost $\costMDP(s,\delta)$, where $\delta \in \Delta$, is the expected communication cost from the waypoint $\wp$ corresponding to MDP state $s$, to the next waypoint, and is given by
\begin{equation}
    \costMDP(s,\delta) = \bar{k}(\wp,\delta) \, \gamma(\delta),
\end{equation}
where $\bar{k}(\wp,\delta)$ is the expected number of time steps required to navigate from $\wp$ to the the next waypoints under $\delta$. 
Further, we include three states $s_\coll$, $s_\tar$, and $s_{\mathrm{free}}$ in addition to the belief states in $S$ to represent termination in an obstacle, target region, and neither, respectively. 

Below, we detail two methods to construct the MDP states: enforced KF convergence and discretized belief states.
\begin{figure}
    \centering
    \includegraphics[width=0.48\textwidth]{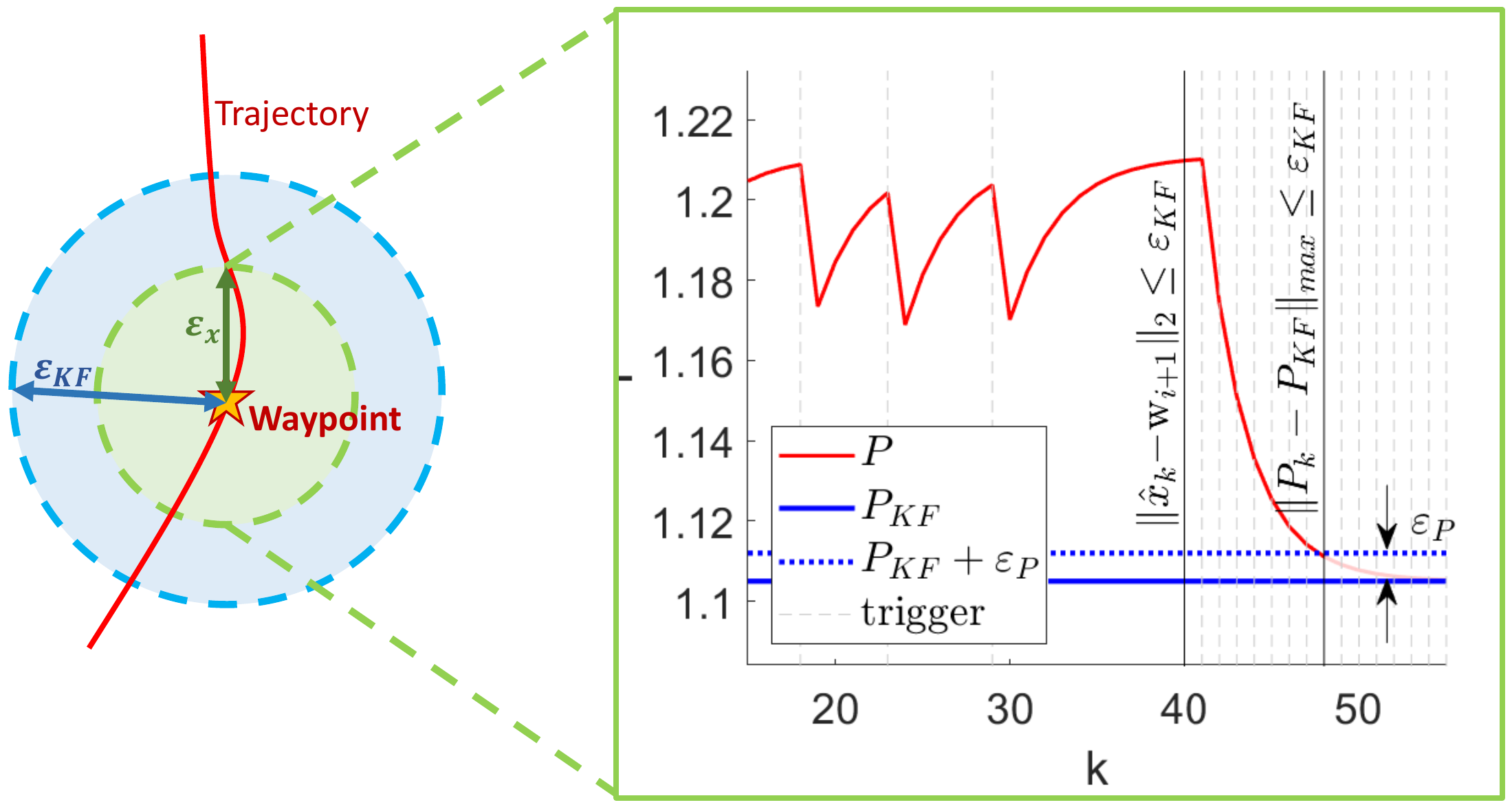}
    \caption{Diagram of Belief Convergence method}
    \label{fig:beliefConvergence}
    \vspace{0mm}
\end{figure}

\begin{figure*}
    \centering
    \begin{subfigure}{0.19\textwidth}
    \includegraphics[width=\textwidth]{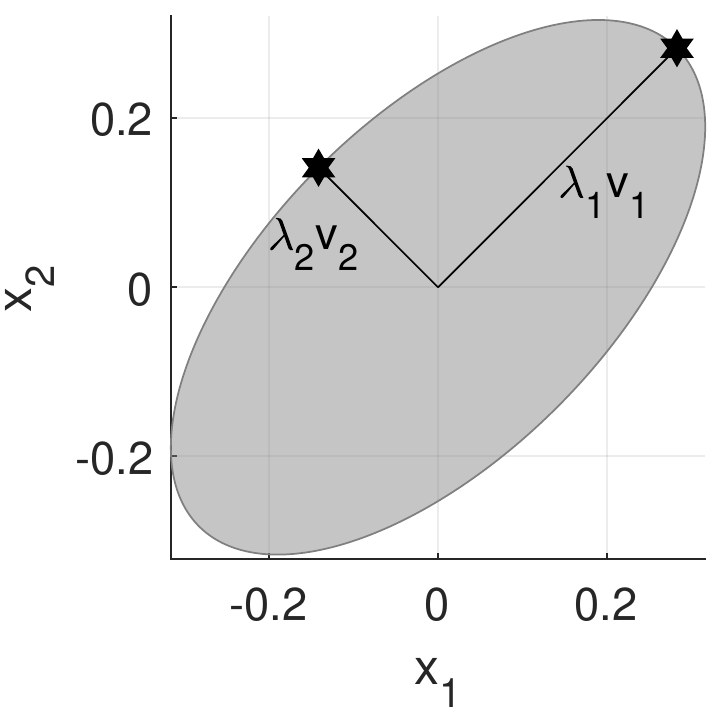}
    \caption{}
    \label{fig:covEll}
    \end{subfigure}
    \hfill
    \begin{subfigure}{0.19\textwidth}
    \centering
    \includegraphics[width=\textwidth]{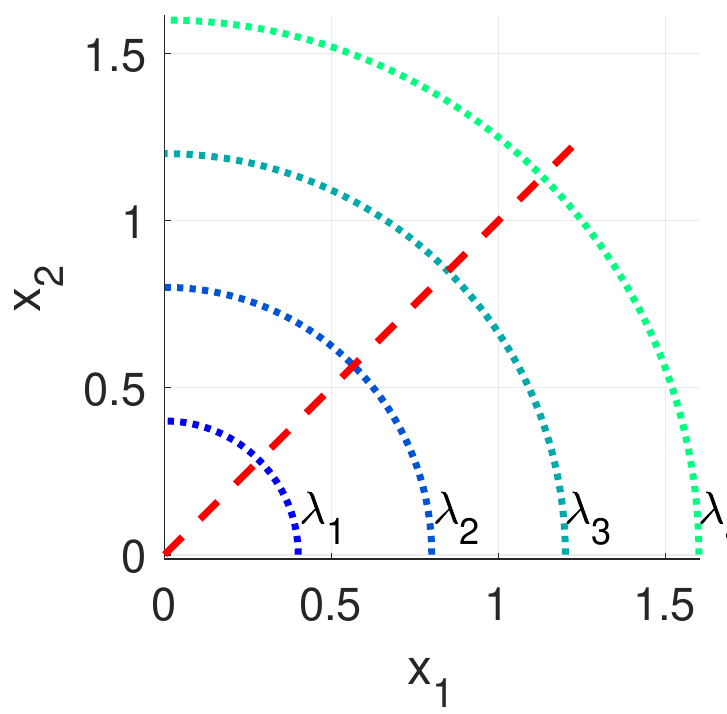}
    \caption{}
    \label{fig:DiscEllAngles}
    \end{subfigure}
    \hfill
    \begin{subfigure}{0.19\textwidth}
    \includegraphics[width=\textwidth]{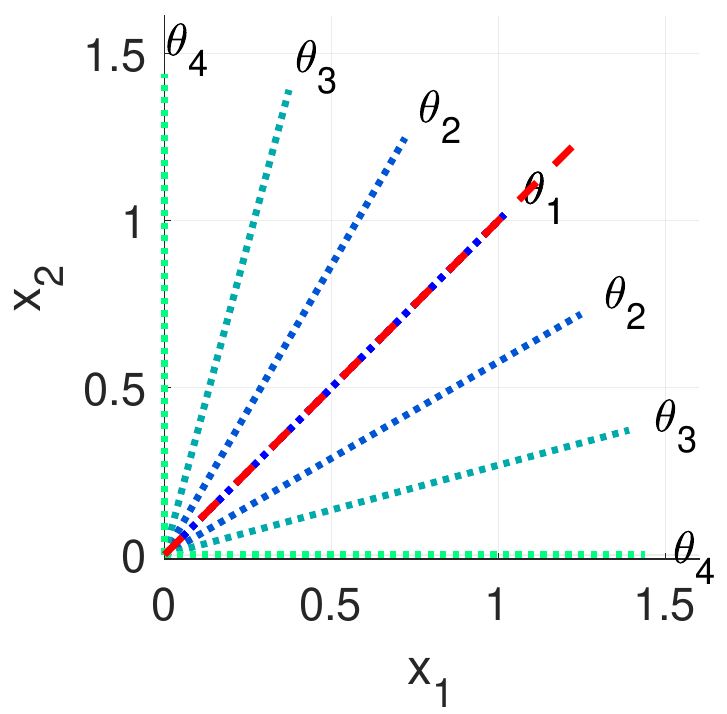}
    \caption{}
    \label{fig:DiscEllMagnitudes}
    \end{subfigure}
    \hfill
    \begin{subfigure}{0.19\textwidth}
    \includegraphics[width=\textwidth]{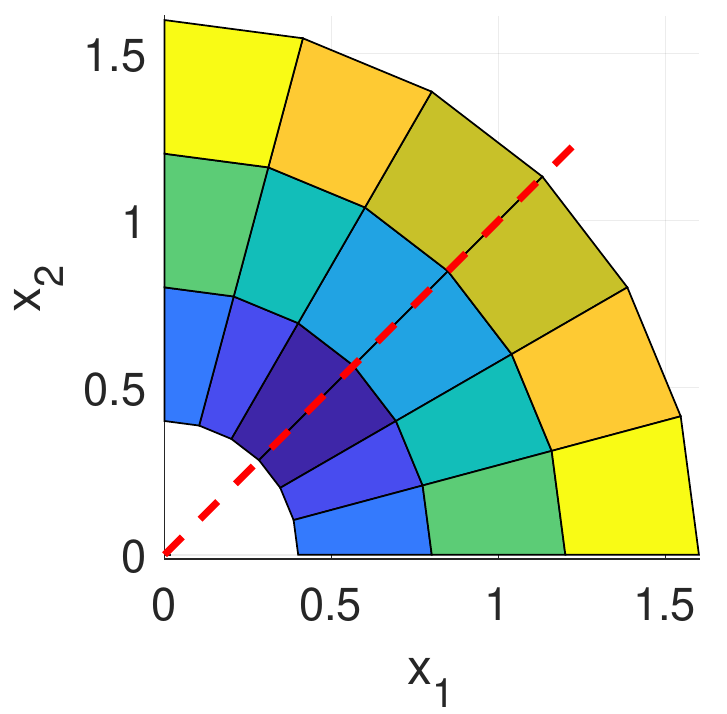}
    \caption{}
    \label{fig:DiscEllRegions}
    \end{subfigure}
    \hfill
    \begin{subfigure}{0.19\textwidth}
    \includegraphics[width=\textwidth]{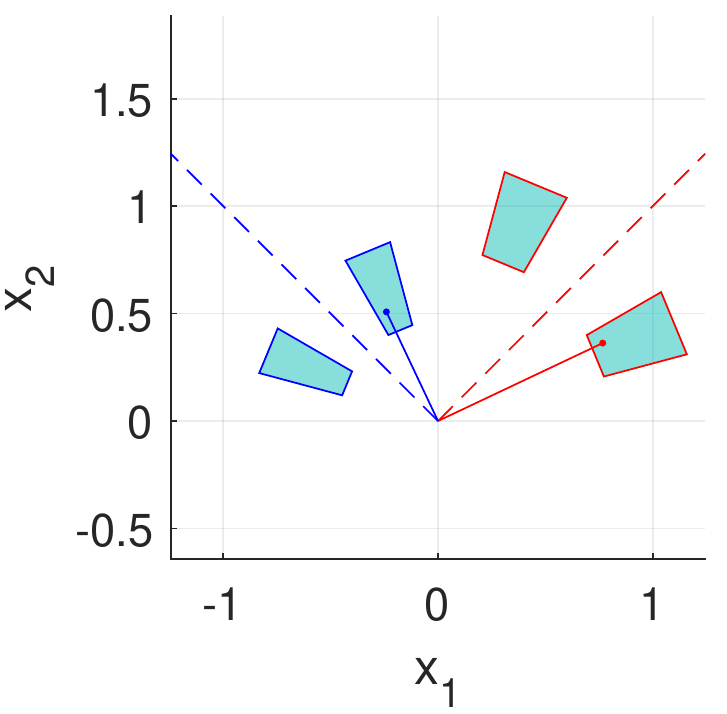}
    \caption{}
    \label{fig:DiscEllTest}
    \end{subfigure}
\caption{(a) Example 2D covariance ellipse decomposed via eigenvalues $\lambda_i$ and eigenvectors $v_i$. (b) Magnitude discretization of a single axis, with $v_{nom}$ shown as a red dashed line.  (c) Angle discretization of a single axis, with $v_{nom}$ shown as a red dashed line.  (d) Full region discretization of a single axis. 
(e) Sample covariance contained within a state defined by a region in the first axis (shown in red and discretized as in c.), and a region in the second axis (shown in blue).}

\end{figure*}
\subsection{Method 1: Enforced Belief Convergence} 

This method relies on convergence to a pre-computed belief to break the history dependence of the belief graph discussed above.
This pre-computed belief at waypoint $\wp$ is a Gaussian distribution defined by the KF steady state covariance $P_\text{KF} \in \mathbb{R}^{n \times n}$, i.e., $\mathcal{N}(\wp, P_\text{KF})$.

To ensure convergence to this belief, we force the system to switch from ET to KF estimation within some neighborhood of $\wp$. 
Let $\tol_\text{KF}\in\mathbb{R}_{>0}$ be the radius of this neighborhood.   
The system switches to KF when $\|\hat{x}_k - \wp_{i+1}\|_2 \leq \tol_{KF}$.  Then, it continues under KF until both the covariance convergence condition, $\|P_k - P_{KF}\|_{max} \leq \tol_P$, where $\tol_P \in \mathbb{R}_{>0}$, and \eqref{eq:meanConvCond}, are met.  
See Figure \ref{fig:beliefConvergence} for reference.
The belief state at each waypoint is thus known to be Gaussian with a covariance within $\tol_P$ of $P_{KF}$ and a mean within $\tol_x$ of the waypoint. So, the MDP state under the Enforced KF Convergence Method is defined as a tuple: 
$s = (\wp, \tol_x, P_{KF}, \tol_P).$

Note that, under every choice of $\delta$,  the system converges to the same MDP belief state at each waypoint. Hence, one of the benefits of this method is its low complexity: the total number of states is $N+3$, where $N$ is the number of waypoints.  It however suffers from constraining the system and does not fully exploit the benefits of ET estimation.

\subsection{Method 2: Discretized Belief State}
The second method operates only in the ET mode
and uses only the mean convergence criterion in \eqref{eq:meanConvCond} for termination. 
This means $P_k$ at each waypoint can be quite variable. As there are two possible cases for the $P_k$ update (implicit and explicit information updates), stochastic $y_k$ induces random switching between these two cases. This switching precludes convergence to steady state $P_k$; even under the same $\delta$, two trajectories may arrive at the same waypoint with drastically different $P_k$. 
We form the MDP states from sets of $b_k$ at each waypoint with $\hat{x}_k$ captured by criterion in \eqref{eq:meanConvCond} and $P_k$ captured by a discretized covariance state (described below). The obtained MDP using this grouping of beliefs leads to history-independent state transitions only up to the resolution used for the covariance discretization, called \textit{abstraction error}.
In Sec. \ref{subsec:DiscErr}, we show that the abstraction error goes to zero as the discretization size goes to zero.

The simplest approach is to directly discretize the individual elements of $P_k$. However, this approach is inefficient and exacerbates the state explosion problem. Instead, we rely on a geometric interpretation of $P_k$: 
the spectral decomposition of $P_k$ results in $n$ eigenvectors and eigenvalues, which respectively describe the orientation axes and magnitudes of an associated uncertainty (hyper)ellipse, as shown in Figure \ref{fig:covEll}. 
The covariance state is then taken as a combination of discrete \emph{regions} across each axis. 


A discrete region for each axis (eigenvector) is defined by 3 elements: a nominal vector, an angle range, and a magnitude range. The nominal vector, $v_{nom} \in \mathbb{R}^n$, is a unit vector that provides the orientation of the region. While the nominal vector can be chosen as any vector in $\mathbb{R}^n$, a good choice is to base it on empirically sampled data. The angle range between $\theta_{low}, \theta_{high} \in (0,\pi/2]$, is defined \emph{relative} to the $v_{nom}$, with the angle calculated via the Euclidean inner product, (various angle choices are shown in \ref{fig:DiscEllAngles}). The magnitude range is between $\lambda_{low}, \lambda_{high} \in \reals_{>0}$ (various magnitude choices are shown in Figure \ref{fig:DiscEllMagnitudes}). A region described by $v_{nom}$, angles $\theta_{low}$ and $\theta_{high}$, and magnitudes $\lambda_{low}$ and $\lambda_{high}$, is defined as the set 
\begin{equation}
    R = \{(\lambda, \theta) \mid \lambda \in [\lambda_{low}, \lambda_{high}) \wedge \theta \in [\theta_{low}, \theta_{high})\},
\end{equation}
for eigenvalue $\lambda$ and angle $\theta$ from the eigenvector to $v_{nom}$ (the combinations of angle and magnitude ranges to form
regions are shown as colored areas in Figure \ref{fig:DiscEllRegions}).

The full covariance state, $c_R$ is taken as a combination of regions for each dimension: $c_R = (R_1, ..., R_n)$. Figure \ref{fig:DiscEllTest} shows the decomposition of an example covariance, and the covariance state it falls in. The MDP state under the Discretized Belief State Method is defined as a tuple: $s = (\wp, \tol_x, c_R)$.

Defining the states as described relies on some knowledge of the upper and lower bounds of the covariance. While the lower bound can be trivially obtained via the KF steady state covariance, we currently have no method of determining the theoretical upper bound for the covariance. Although this is a potential avenue for future work, a good empirical bound can be obtained by propagating the covariance through the update equation and enforcing $\gamma_k=0$ at every time step. 

\subsubsection{Discretization Error}
\label{subsec:DiscErr}
Modeling errors arise whenever a continuous state space dynamical system is spatially discretized. In the MDP, this error presents in the transition probabilities: belief states originating from some parts of the same region (as defined above) have different transition probabilities from those originating in different parts of the same region. This discrepancy in the probabilities is the abstraction error. Theorem \ref{theorem}  shows that as the volume of the discretized belief state, i.e. the size of the ranges $(\theta_{low}, \theta_{high})$ and $(\lambda_{low}, \lambda_{high})$, goes to zero, the abstraction error goes to zero.

\begin{theorem}\label{theorem}
Consider the transition from a set of belief states $\mathcal{S}$ to another set $\mathcal{S}'$. The volume of each set is defined by the discretization sizes: $d\lambda=\lambda_{high}-\lambda_{low}$, $d\theta=\theta_{high}-\theta_{low}$. Assume $d\lambda$ and $d\theta$ are equivalent for $\mathcal{S}$ and $\mathcal{S}'$. Every point in the originating set, $s \in \mathcal{S}$, has an associated transition probability to $\mathcal{S}'$, $P(\mathcal{S}'|s)$. The maximum difference between transition probabilities within the set is:
\begin{equation}
    \Delta P_{max}=\max_{s_1,s_2\in \mathcal{S}} |P(\mathcal{S}'|s_1) - P(\mathcal{S}'|s_2)|.
\end{equation}
As the discretization step sizes go to zero, this difference must also go to zero:
\begin{equation}
    \lim_{d\lambda \to 0, d\theta\to 0} \Delta P_{max} = 0.
\end{equation}

\end{theorem}

\begin{proof}

Within the set $\mathcal{S}$, there must exist some point, $s_{max}$, that maximizes the transition probability, and some point, $s_{min}$, that minimizes this probability:
\begin{equation}
    s_{max}=\max_{s\in \mathcal{S}} P(\mathcal{S}'|s), \quad s_{min}=\min_{s\in \mathcal{S}} P(\mathcal{S}'|s).
\end{equation}

Now, consider a new set that is strictly a subset of $\mathcal{S}$, $\tilde{\mathcal{S}}\subset \mathcal{S}$. Because it is a subset of $\mathcal{S}$, we can use $\tilde{\mathcal{S}}$ to bound the probabilities:
\begin{align}
\max_{s\in \mathcal{S}} P(\mathcal{S}'|s) \geq \max_{s\in \tilde{\mathcal{S}}} P(\mathcal{S}'|s) \\
\min_{s\in \mathcal{S}} P(\mathcal{S}'|s) \leq \min_{s\in \tilde{\mathcal{S}}} P(\mathcal{S}'|s).
\end{align}
In the limit as the volume of $\mathcal{S}$ goes to zero, these bounds must converge to each other. Thus, the difference in probabilities within the volume, $\Delta P_{max}$,  must also go to zero.
\end{proof}

\subsubsection{Scalability}
A simpler method of covariance discretization could involve setting ranges on the individual elements of the covariance matrix, rather than relying on the spectral decomposition. While attractive in its simplicity, this alternative exacerbates an existing problem with discretization: state explosion. Consider an analogous definition of a region as some range of values over a single element of the covariance matrix. The full MDP state could then be described as a combination of regions across each element of the covariance matrix. Since the covariance matrix is symmetric, the number of regions, $N_R$, required to describe a state is
\begin{equation}
    N_R = \frac{n(n+1)}{2},
\end{equation}
where $n$ is the dimension of the state space. Whereas the spectral decomposition requires only a single region for each axis, so the number of regions required to describe a state is simply $n$.

Now assume that each matrix element is discretized into $d$ regions, and similarly that each axis in the spectral decomposition is also discretized into $d$ regions. The total number of combinations across all the regions is $d^{N_R}$. While this number explodes for both cases, it explodes less rapidly for the discretization based on the spectral decomposition. If we assume the same discretization at every waypoint, the total number of states is then $Nd^{N_R}+3$. While this is a worst case scenario (in practice, not every combination of regions is occupied), the described techniques do suffer from a state explosion problem, which is particularly troublesome as the dimensionality of the system increases.

The number of required regions for the matrix elements could be reduced by adding constraints inherent to the covariance, like positive definiteness and orthogonality. However, interpreting these constraints in matrix form is not straightforward, while they are very simple to implement in a spectral decomposition. The spectral decomposition also benefits from visual interpretability; it is easy to visualize the uncertainty of the system in terms of vectors defining a covariance ellipse. 

Off the shelf tools are available to solve the multi-objective problem on an MDP \cite{Kwiatkowska2011}. These use value iteration to construct a Pareto Front \cite{Forejt2012}, from which specific strategies are generated through linear programming \cite{Etessami2007}. The MDP analysis algorithms are polynomial with the size of the MDP.

For both methods, we use Monte Carlo sampling to approximate the transition probabilities \cite{Lahijanian2018}.

\section{EVALUATIONS}
    \label{sec:eval}
    
In this section we demonstrate the benefits of trade-off analysis by showing that small sacrifices in task performance can lead to large energy savings. We present Pareto Fronts for a 2D and 3D system generated using PRISM-games, \cite{Kwiatkowska2013}, and simulate specific strategies for interesting points. We simulated multiple points from each system and trajectory to compare the theoretical values for the objectives (obtained through analysis on the MDP) to the empirical results for the objectives (from 3000 simulations of the strategies). In all cases for the 2D system, we found the error between theory and simulation to be less than 2\%, and for the 3D system we found the error to be less than 3.5\%. 

\subsection{2D System}
We consider the system from \cite{Bry2011}, with dynamics:
\begin{align}
\label{eq:2Ddynamics}
    x_{k+1} &= x_k + u_k + w_k \\
    y_k &= x_k + v_k
\end{align}
with process noise $w_k \sim \mathcal{N}(0,0.07^2I)$ and measurement noise $v_k \sim \mathcal{N}(0,0.03^2I)$. For clarity, the communication cost is $c_m=1$ so that the energy is simply the raw number of triggers.

We first consider a winding trajectory through a set of obstacles, shown in Figure \ref{superGPP1}. For this trajectory, we built an MDP using both of the two abstraction methods described above. The resulting Pareto fronts are shown in Figure \ref{PPfrontCompMethods}. Note that due to the size of the MDP in the Discretized Belief Method, the Pareto Front is very complex, and Figure \ref{PPfrontCompMethods}.a. shows only a selection of interesting points. On the other hand, the simplicity of the MDP for the Enforced Belief Convergence Method allows for fast and easy computation of all vertices describing the full front, shown in Figure \ref{PPfrontCompMethods}.b..

\begin{figure}[ht]
    \centering
    \begin{subfigure}{0.23\textwidth}
    \centering
        \includegraphics[width=\textwidth]{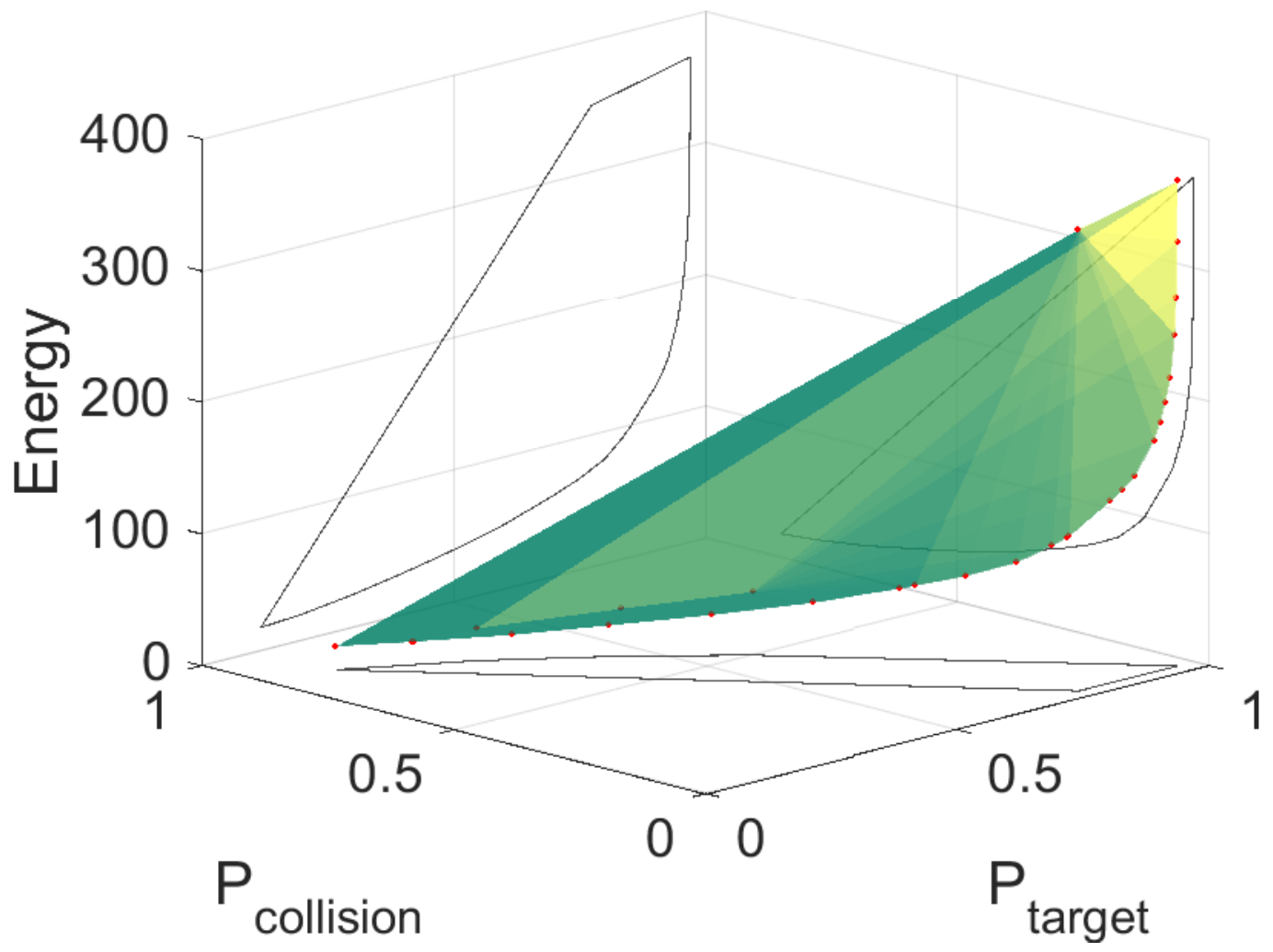}
        \caption{Discretized Belief Method}
    \end{subfigure}
    \hfill
    \begin{subfigure}{0.23\textwidth}
    \centering
        \includegraphics[width=\textwidth]{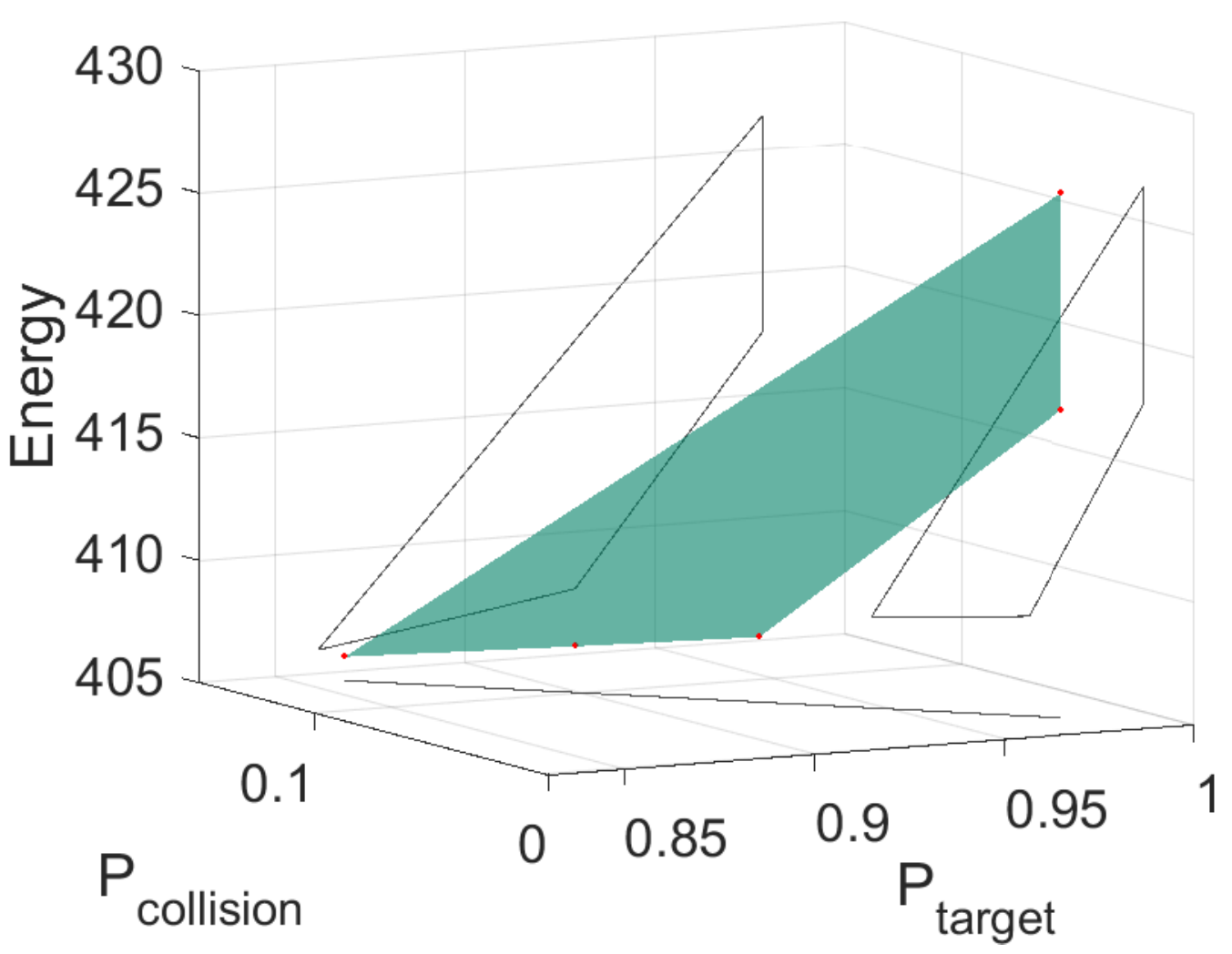}
        \caption{Enforced Belief Convergence Method}
    \end{subfigure}
    \caption{Pareto Fronts for 2D system with winding trajectory.}
    \label{PPfrontCompMethods}
\end{figure}

Each front shows a characteristic shape; from the highest probability of target there is a steep drop-off in energy which eases as the probability of target gets lower. This indicates relatively large energy savings for the highest probability of target, with smaller gains as the probability of target goes to its minimum. This shape is common across every Pareto Front generated in these studies, and supports the argument that this kind of analysis can provide valuable energy savings for little sacrifice in performance objectives. 

The second note from these two fronts is the contrast between the two methods. At the highest energy consumption, both methods offer similarly high probability of reaching the target ($98.8\%$ for the KF Convergence Method vs. $96.8\%$ for the Discretized Belief Method). However, the discretized belief method drops off to far lower values than the KF convergence method, indicating that method is better able to leverage ET and provide greater energy efficiency. 

\begin{table}[h!]

    \begin{center}
    \caption{3 Pareto Points for winding trajectory using the Discretized Belief Method}
    \label{tab:2Dwinding}
    \begin{tabular}{|| c | c | c | c ||}
    \hline
    Pareto Point & $P_{tar}$ & $P_{col}$ & $E_c$ \\
    \hline
         1 & 96.85\% & 3.14\% & 367.99 \\
         2 & 95\% & 4.70\% & 185.07 \\
         3 & 83.5\% & 14.68\% & 92.35 \\
         Full KF & 95.67\% & 3.97\% & 682.32\\
    \hline
    \end{tabular}
    \end{center}
\end{table}

To more closely examine interesting trade-offs for the Discretized Belief Method, consider Table \ref{tab:2Dwinding}, which shows 3 selected Pareto Points of interest in addition to expected performance if a regular KF is used at every time step. Pareto Point 1 is the point on the front with the highest probability of reaching the target. Pareto point 2 sacrifices only $1.8\%$ probability of target, but uses half the energy. However, Pareto Point 3 represents another halving of the energy, but the probability of target drops by $11.5\%$: the strength of the trade-off decreases as the energies become lower. Also note the stark comparison with the full KF performance; at its highest probability of reaching the target, ET achieves the similar performance for $54\%$ the energy use. The behaviors of the three Pareto Points can be seen in the simulation of their corresponding strategies, shown in Figure \ref{TrajectoryComp}.

\begin{figure}[b]
     \centering
     \begin{subfigure}[b]{0.13\textwidth}
         \centering
         \includegraphics[width=\textwidth]{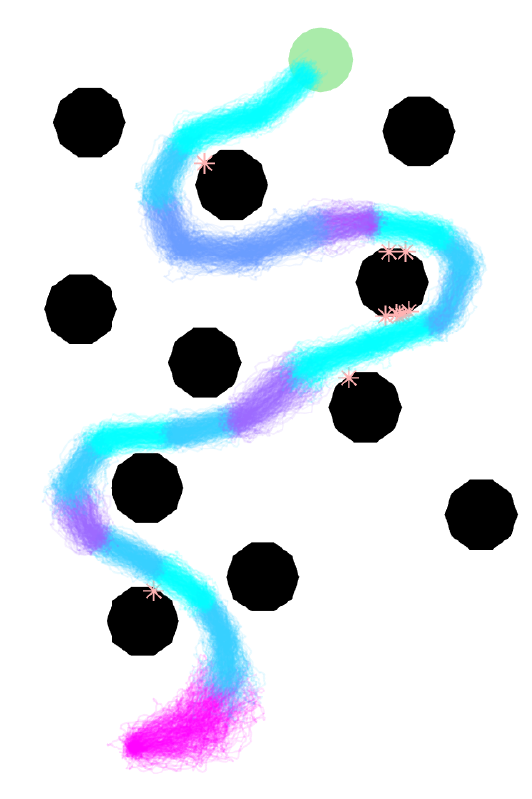}
         \caption{Pareto Point 1}
         \label{superGPP1}
     \end{subfigure}
     \hfill
     \begin{subfigure}[b]{0.13\textwidth}
         \centering
         \includegraphics[width=\textwidth]{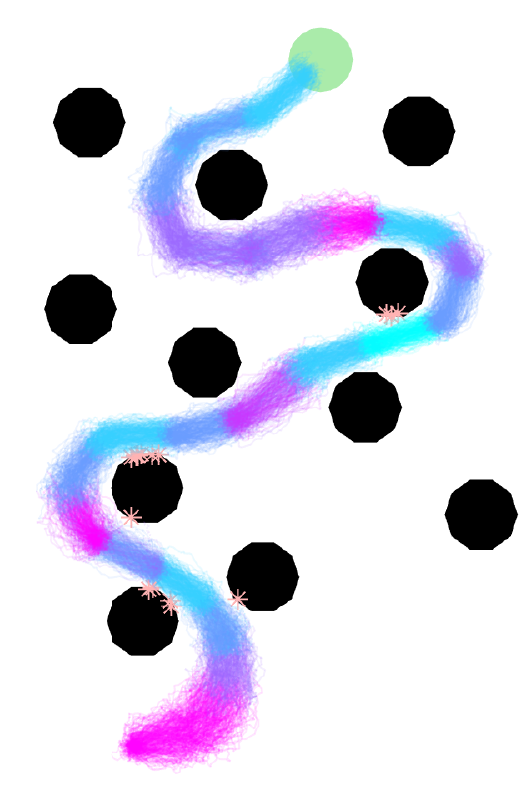}
         \caption{Pareto Point 2}
         \label{superGPP2}
     \end{subfigure}
     \hfill
     \begin{subfigure}[b]{0.13\textwidth}
         \centering
         \includegraphics[width=\textwidth]{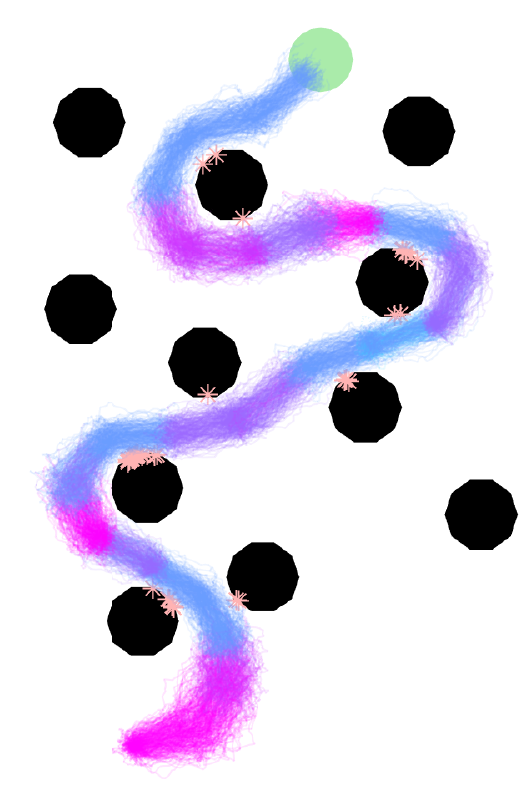}
         \caption{Pareto Point 3}
         \label{superGPP3}
     \end{subfigure}
          \hfill
     \begin{subfigure}[b]{0.07\textwidth}
         \centering
         \includegraphics[width=\textwidth]{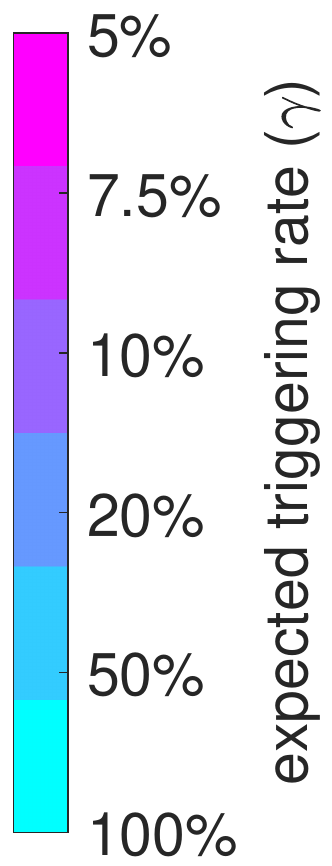}
         \label{superGPP3}
     \end{subfigure}
 
    \caption{Simulated strategies of three Pareto points for a 2D system using the Discretized Belief method. The color of the trajectory indicates $\bar{\gamma}(\delta)$ 
    for the $\delta$ chosen at each waypoint.}
    \label{TrajectoryComp}
\end{figure}

Figure \ref{TrajectoryComp} shows 300 simulated runs of each strategy, where the color of the trajectory indicates the triggering threshold. Each threshold can be directly mapped to an expected triggering rate (see \cite{Wu2013}). In the figure, cyan corresponds the highest expected triggering rate (100\%) and magenta the lowest (5\%). The efficiency of the strategies can be seen in how the triggering rate only increases (lower $\delta$) in proximity to obstacles, and otherwise relaxes (higher $\delta$) in more open parts of the trajectory. Higher triggering rates lead to a contraction in the trajectories, and therefore lower probabilities of collision. Trivially, we can see that the lower energy Pareto points correspond to lower triggering rates around the obstacles.

\begin{figure}
     \centering
     \begin{subfigure}[b]{0.3\textwidth}
         \centering
         \includegraphics[width=\textwidth]{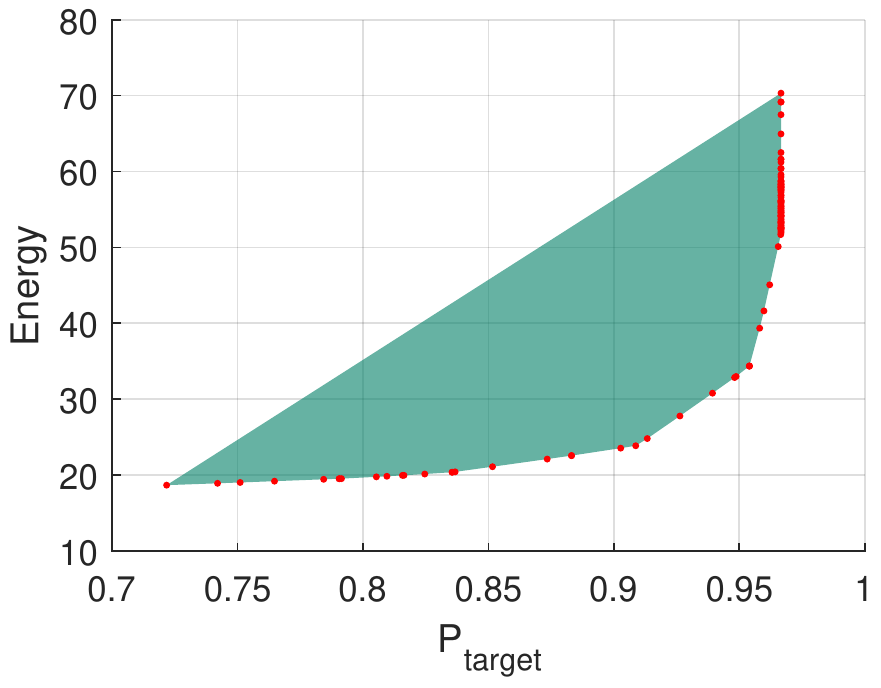}
         \caption{Pareto Front}
         \label{openPfront}
     \end{subfigure}
     \hfill
     \begin{subfigure}[b]{0.2\textwidth}
         \centering
         \includegraphics[width=\textwidth]{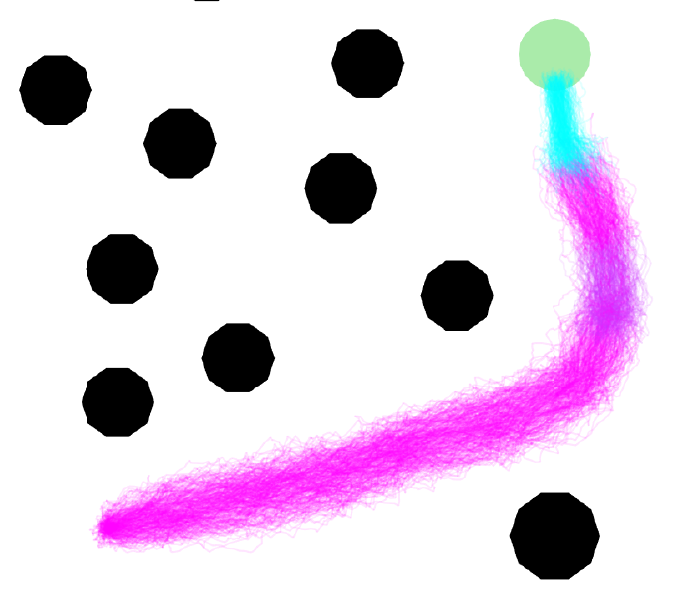}
         \caption{Pareto Point 1}
         \label{openPP1}
     \end{subfigure}
     \hfill
     \begin{subfigure}[b]{0.2\textwidth}
         \centering
         \includegraphics[width=\textwidth]{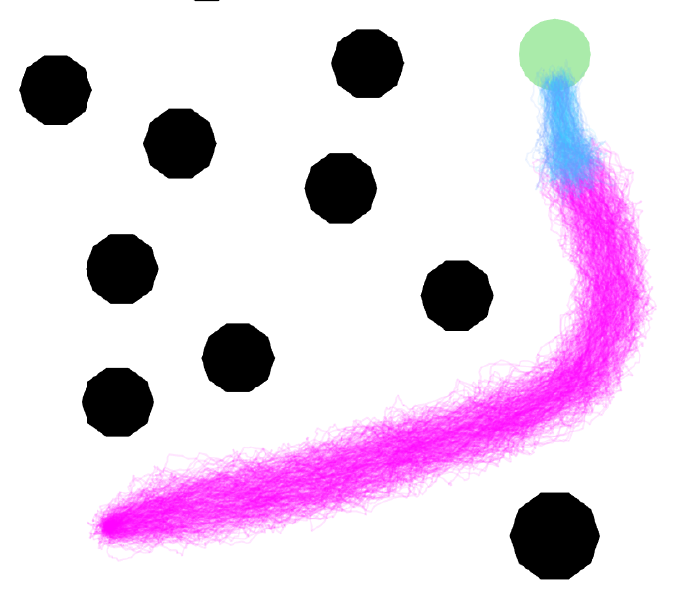}
         \caption{Pareto Point 2}
         \label{openPP2}
     \end{subfigure}
     \hfill
     \begin{subfigure}[b]{0.07\textwidth}
         \centering
         \includegraphics[width=\textwidth]{figures/colorbar.pdf}
     \end{subfigure}
        \caption{Pareto front and 2 simulated strategies for 2D System, open trajectory, using discretized belief method.}
        \label{openTrajectory}
\end{figure}

The Pareto front is specific to the trajectory chosen, and different trajectories yield different trade-offs. To demonstrate this, we consider a more open trajectory for the same 2D system with the Discretized Belief Method. This trajectory does not collide with any obstacles. The resulting Pareto Front and trajectories for two simulated points are shown in Figure \ref{openTrajectory}. Performance of the objectives for the two simulated Pareto points is shown in Table \ref{tab:2Dopen}.
\begin{table}[h!]
\begin{center}
    \caption{Two Pareto Points for open trajectory with Discretized Belief Method}
    \label{tab:2Dopen}
    \begin{tabular}{|| c | c | c | c ||}
    \hline
    Pareto Point & $P_{tar}$ & $P_{col}$ & $E_c$ \\
    \hline
         1 & 97.66\% & 0.0 & 52.58 \\
         2 & 95\% & 0.0 & 30.95\\
         Full KF & 97.43\% & 0.0 & 332.82\\
    \hline
    \end{tabular}
\end{center}
\end{table}

Note that because this trajectory avoids proximity with obstacles, the only difference in the Pareto Points comes at the very end of the trajectory going into the target region. Because there is really only one waypoint with a relevant action choice, the trade-offs between Pareto points are far less dramatic than the first trajectory, however the contrast with a using a full KF at every time step is amplified.

\subsection{3D System}
Finally, to demonstrate the flexibility of our abstraction method, we consider the same system as (\ref{eq:2Ddynamics}), but extended to three dimensions. A selection of interesting Pareto Points and a single simulated strategy for a winding trajectory using the Discretized Belief Method is shown in Figure \ref{3DTrajectory}. Note that the same general shape of the Pareto front holds for the 3D system, with generous trade-offs at higher energies that taper as the energy decreases. The simulated trajectory corresponds to the highest probability of target and shows similar behavior to the 2D system, with higher triggering rates closer to the obstacles.

\begin{figure}[t]
     \centering
     \begin{subfigure}[b]{0.23\textwidth}
         \centering
         \includegraphics[width=\textwidth]{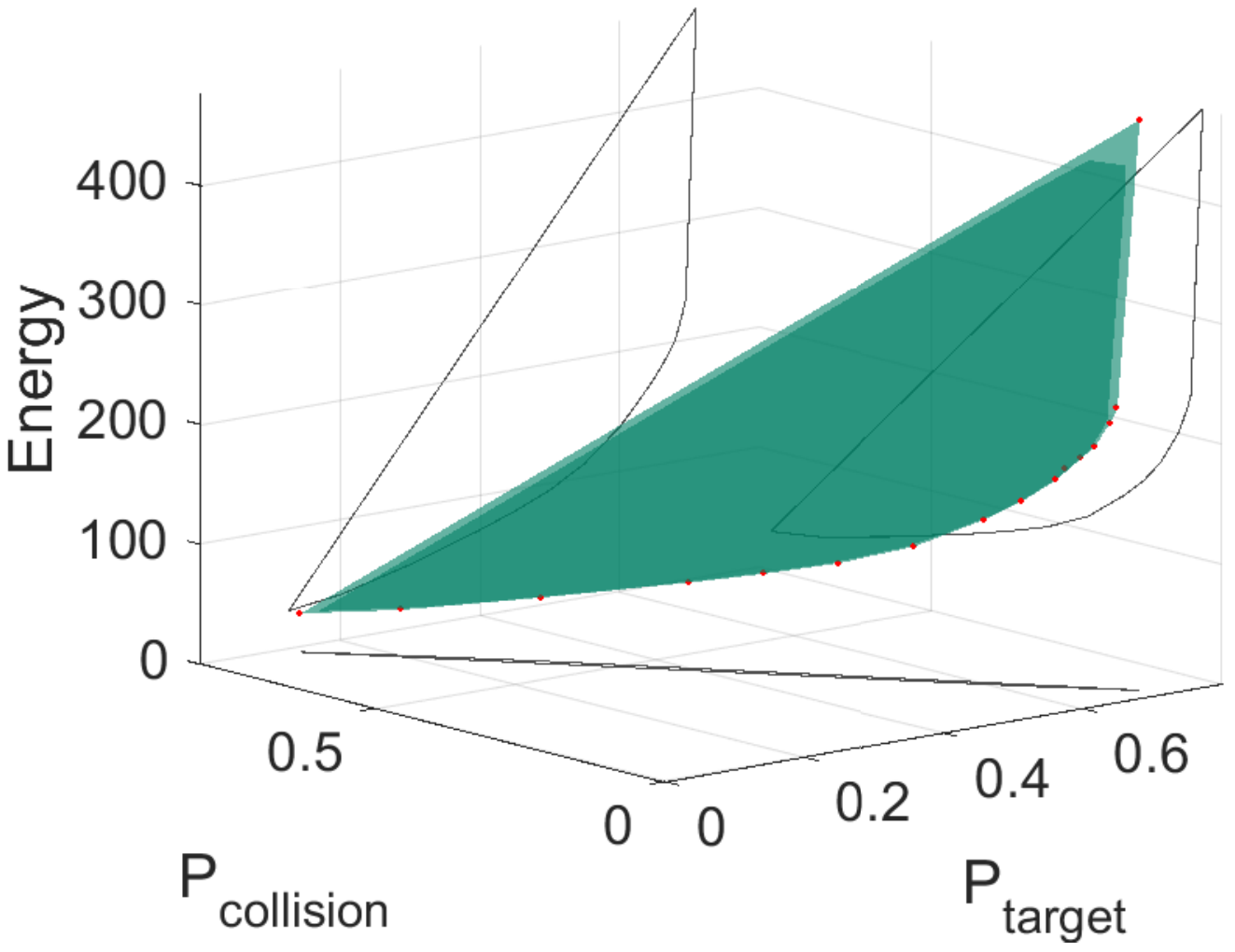}
         \caption{Pareto Front}
         \label{3DPfront}
     \end{subfigure}
     \hfill
     \begin{subfigure}[b]{0.23\textwidth}
         \centering
         \includegraphics[width=\textwidth]{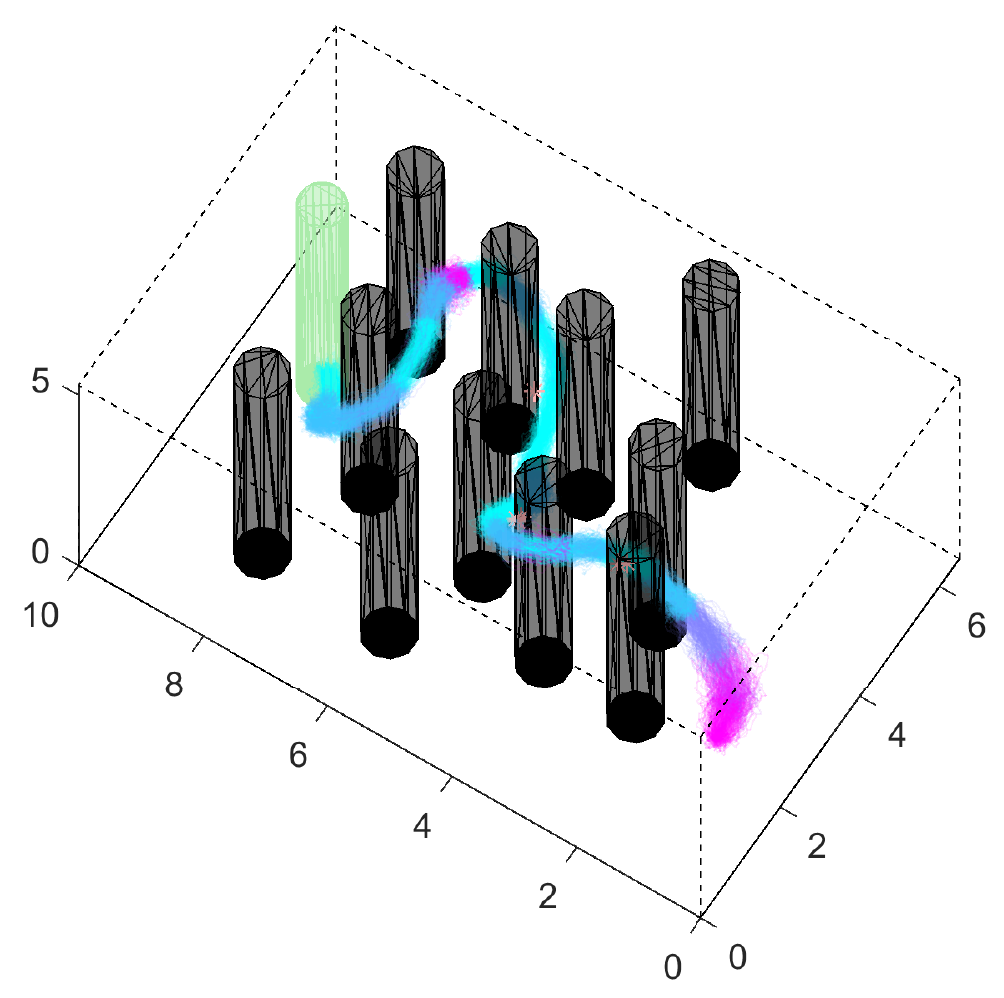}
         \caption{Pareto Point 1}
         \label{3DPP1}
     \end{subfigure}

        \caption{(a) Pareto front and simulated strategy for 3D System; (b) `Super G' trajectory, using discretized belief method. } 
        \label{3DTrajectory}
\end{figure}


\section{CONCLUSION}
    \label{sec:conclusion}



In this work, we introduced a method of calculating optimal trade-offs between communication cost of a distributed sensor network via ET, and task performance of an active agent. Our method relies on a novel geometric discretization using the spectral decomposition of a covariance matrix to group belief states of similar size and orientation. This technique allows us to fully leverage the benefits offered by ET. A limitation of this work is scalability. A future research direction is to employ adaptive discretization techniques to gain efficiency, extending the method to higher dimensions.

\bibliographystyle{IEEEtran}
\bibliography{refs}


\end{document}